%% file: Article.tex
\documentclass[submission,copyright,creativecommons,noderivs,noncommercial]{eptcs}
\providecommand{\event}{GandALF 2017} 
\usepackage{underscore}           

\usepackage[table]{xcolor}
\usepackage{wrapfig}
\usepackage{subfigure}
\usepackage{hhline}
\usepackage{multirow}
\usepackage{paralist}
\usepackage{tikz}
\usetikzlibrary{positioning,arrows}
\usepackage{multicol}
\usepackage{mathtools}
\usepackage{xspace}

\usepackage{algorithm}
\usepackage[noend]{algpseudocode}
\algrenewcommand{\algorithmiccomment}[1]{\hfill$\triangleleft$ {\footnotesize \textsl{#1}}}
\algloopdefx{Return}[1]{\textbf{return} #1}

\algrenewcommand\algorithmicindent{1em}
\makeatletter
\renewcommand{\ALG@beginalgorithmic}{\small}
\makeatother

\usepackage{amsmath,amsfonts,amssymb,amsthm}

\newcommand \tpl[1]{\langle #1 \rangle}

\newcommand \details[1]{}
\newcommand{\ATLStar}{\text{\sffamily ATL$^{*}$}}
\newcommand{\CTLStar}{\text{\sffamily CTL$^{*}$}}
\newcommand{\ATL}{\text{\sffamily ATL}}
\newcommand{\ATA}{\text{\sffamily ATA}}
\newcommand{\ACG}{\text{\sffamily ACG}}

\newcommand{\NTA}{\text{\sffamily NTA}}
\newcommand{\CTL}{\text{\sffamily CTL}}
\newcommand{\LTL}{\text{\sffamily LTL}}
\newcommand{\CGS}{\text{\sffamily CGS}}
\newcommand{\CGT}{\text{\sffamily CGT}}

\newcommand{\NWA}{\text{\sffamily NWA}}

\newcommand{\until}{\textsf{U}}

\newcommand{\Next}{\textsf{X}}
\newcommand{\A}{\textsf{A}}
\newcommand{\E}{\textsf{E}}
\newcommand{\Always}{\textsf{G}}
\newcommand{\Eventually}{\textsf{F}}
\newcommand{\Exists}[1]{\langle\hspace{-0.05cm}\langle #1 \rangle\hspace{-0.05cm}\rangle}
\newcommand{\Forall}[1]{[\hspace{-0.05cm}[ #1 ]\hspace{-0.05cm}]}

\newcommand{\THREEXPTIME}{{\sc 3Exptime}\xspace}
\newcommand{\TWOEXPTIME}{{\sc 2Exptime}\xspace}
\newcommand{\TWOEXPSPACE}{{\sc 2Expspace}\xspace}
\newcommand{\EXPTIME}{{\sc Exptime}\xspace}
\newcommand{\PTIME}{{\sc Ptime}\xspace}
\newcommand{\PSPACE}{{\sc Pspace}\xspace}

\newcommand{\Prop}{{\textit{AP}}}
\newcommand{\Agents}{{\textit{Ag}}}
\newcommand{\agent}{{\textit{a}}}
\newcommand{\Actions}{{\textit{Ac}}}
\newcommand{\Decisions}{{\textit{Dc}}}
\newcommand{\decision}{{\textit{d}}}

\newcommand{\dir}{{\textit{dir}}}
\newcommand{\States}{{\textit{S}}}
\newcommand{\Trans}{{\tau}}
\newcommand{\Lab}{{\textit{Lab}}}
\newcommand{\outcome}{{\textit{out}}}
\newcommand{\Unw}{{\textit{Unw}}}
\newcommand{\Trk}{{\textit{Trk}}}
\newcommand{\GS}{{\mathcal{G}}}
\newcommand{\Au}{{\mathcal{A}}}

\newcommand{\env}{{\textit{env}}}

\newcommand{\sys}{{\textit{sys}}}
\newcommand{\exec}{{\textit{exec}}}

\newcommand{\Atoms}{{\textit{Atoms}}}

\newcommand{\Tau}{{\mathcal{T}}}

\newcommand{\Lang}{{\mathcal{L}}}

\newcommand{\lst}{{\textit{lst}}}
\def\B{{\mathbb{B}}}
\def\Nat{{\mathbb{N}}}
\newcommand{\true}{\texttt{true}}

\newcommand{\begOne}{{\textit{beg}}}
\newcommand{\begTwo}{{\textit{sb-beg}}}
\newcommand{\checkOne}{{\textit{check}}}
\newcommand{\markTwo}{{\textit{sb-mark}}}
\newcommand{\checkTwo}{{\textit{sb-beg-check}}}
\newcommand{\fair}{{\textit{fair}}}

\newcommand{\EndOne}{{\textit{end}}}
\newcommand{\EndTwo}{{\textit{sb-end}}}
\newcommand{\dire}{{\textit{dir}}}
\newcommand{\Succ}{{\textit{next}}}

\newcommand{\Tag}{{\textit{tag}}}
\newcommand{\WTC}{{\textit{WTC}}}
\newcommand{\TC}{{\textit{TC}}}

\newcommand{\BCT}{{\textit{BCT}}}

\newtheorem{proposition}{Proposition}

\newtheorem{definition}{Definition}
      \newtheorem{theorem}{Theorem}
    \newtheorem{lemma}{Lemma}
    \newtheorem{corollary}{Corollary}
\newtheorem{remark}{Remark}

\title{On the Complexity of  \ATL\ and \ATLStar\ Module Checking}

\author{Laura Bozzelli \qquad Aniello Murano
\institute{University of Napoli ``Federico II'', Napoli, Italy}}

\def\titlerunning{On the Complexity of  \ATL\ and \ATLStar\ Module Checking}
\def\authorrunning{L. Bozzelli  \& A. Murano}

\begin{document}

\maketitle

\begin{abstract}
\emph{Module checking} has been introduced in late 1990s to verify open systems, i.e., systems whose behavior depends on the continuous interaction with the environment.
Classically, module checking has been investigated with respect to specifications given as \CTL\ and \CTLStar\ formulas.
Recently, it has been shown that \CTL\ (resp., \CTLStar) module checking offers a distinctly different perspective from the better-known problem of \ATL\ (resp., \ATLStar) \emph{model} checking. In particular, \ATL\ (resp., \ATLStar) module checking strictly enhances the expressiveness of both \CTL\ (resp., \CTLStar) module checking and \ATL\ (resp. \ATLStar) model checking.
In this paper, we provide asymptotically optimal bounds on the computational cost of module checking against \ATL\ and \ATLStar, whose upper bounds are based on an automata-theoretic approach.
 We show that   module-checking   for \ATL\ is \EXPTIME-complete, which is the same complexity of module checking against \CTL. On the other hand, \ATLStar\ module checking  turns out to be \THREEXPTIME-complete, hence exponentially harder than \CTLStar\ module checking.
\end{abstract}
%

\input{Introduction}

\input{Preliminaries}

\input{DecisionProcedures}

\input{LowerBounds}

\section{Conclusion}\label{sec:conc}


Module checking is a useful game-theoretic framework to deal with branching-time specifications. The setting is simple and powerful as it allows to capture the essence of the adversarial interaction between an open system (possibly consisting of several independent components) and its unpredictable environment. The work on module checking has brought an important contribution to the strategic reasoning field, both in computer science and AI~\cite{AlurHK02}. Recently, \CTL/\CTLStar\ module checking has come to the fore as it has been shown that it is incomparable with \ATL/\ATLStar\ model checking~\cite{JamrogaM14}. In particular the former can keep track of all moves made in the past, while the latter cannot. This is a severe limitation in \ATL/\ATLStar\ and has been studied under the name of irrevocability of strategies in~\cite{AGJ07}. Remarkably, this feature can be handled with more sophisticated logics such as \emph{Strategy Logics}~\cite{CHP10,MMPV14}, \emph{\ATL\ with strategy contexts}~\cite{LaroussinieM15}, and \emph{quantified} \CTL\ \cite{LaroussinieM14}. 
 However, for such logics, the relative model checking question turns out to be non-elementary.

In this paper, we have addressed and carefully investigated the computational complexity of the module-checking problem against \ATL\ and \ATLStar\ specifications. We have shown that \ATL\ module-checking is \EXPTIME-complete, while \ATLStar\ module-checking is \THREEXPTIME-complete. The latter corrects an incorrect claim made in \cite{JamrogaM15}. Note that following~\cite{LaroussinieM15}, \ATLStar\ (resp., \ATL) module-checking can be reduced to model checking against \emph{quantified} \CTLStar\ (resp., \emph{quantified} \CTL), but this approach would lead to non-elementary algorithms for the considered problems.
This work opens to several directions for future work. Mainly, we aim to investigate the same problem in the imperfect information setting as well as for infinite-state open systems.


\bibliographystyle{eptcs}
\bibliography{bib2}

\end{document}

%% file: Introduction.tex
\section{Introduction}\label{sec:intro}

\emph{Model checking} is a well-established formal-method technique to automatically check for global correctness of systems~\cite{Clarke81ctl,Queille81verification}.
In this verification
method, the behavior of a system, formally described by a mathematical model, is
checked against a behavioral constraint specified by a formula in a suitable temporal logic.
Originally, model checking was introduced to analyze finite-state \emph{closed systems} whose dynamic behavior is completely determined by their internal states and transitions. In this specific setting, system models are usually given as labeled-state transition-graphs equipped with some internal degree of nondeterminism (e.g., Kripke structures). An unwinding of the graph results in an infinite tree, properly called \emph{computation tree}, that collects all the possible evolutions of the system. Model checking of a closed system amounts to check whether the computation tree satisfies the specification.
Properties for model checking are usually specified in temporal logics such as \LTL, \CTL, and \CTLStar~\cite{Pnueli77,EmersonH86}, or  alternating-time temporal logics   such as \ATL\ and \ATLStar~\cite{AlurHK02}, the latter ones being extensions of \CTL\ and \CTLStar, respectively, which allow for reasoning about the strategic capabilities of groups of agents.

In the last two decades, interest has arisen in analyzing the behavior of
individual components (or sets of components) in systems with multiple entities.
The interest began in the field of \emph{reactive systems}, which are characterized by a continuous interaction with their (external) environments. One of the first approaches introduced to model check finite-state reactive systems is \emph{module checking}~\cite{KV96}. 
In this setting, the system is modeled as a \emph{module} that
interacts with its environment, and correctness means that a desired property must hold with respect to all possible interactions. Technically speaking, the module is a transition system whose states are partitioned into those controlled by the system and those controlled by the environment. The latter ones intrinsically carry an additional source of nondeterminism describing the possibility that the computation, from these states, can continue with any subset of its possible successor states. This means that while in model checking, we have only one computation tree representing the possible evolution of the system, in module checking we have an infinite number of trees to handle, one for each possible behavior of the environment. Deciding whether a module satisfies a property amounts to check that all such trees satisfy the property. 
This makes the module-checking problem harder to deal with. Indeed, while \CTL\ (resp., \CTLStar) model checking is \PTIME-complete (resp., \PSPACE-complete)~\cite{EmersonH86},
\CTL\ (resp., \CTLStar) module checking is \EXPTIME-complete (resp., \TWOEXPTIME-complete)~\cite{KV96} with a \PTIME-complete complexity for a fixed-size formula.

For a long time, there has been a common belief that module checking of \CTL/\CTLStar\ is a special case of model checking of \ATL/\ATLStar. Because of that, active research on module checking subsided shortly after its conception.
The belief has been recently refuted in~\cite{JamrogaM14}. There, it was proved that module checking includes two features inherently absent in the semantics of \ATL/\ATLStar, namely irrevocability and nondeterminism of strategies.
This result has brought back the interests in module checking as an interesting formalism for the verification of open systems. In particular, in~\cite{JamrogaM14}, several scenarios were discussed to show the usefulness of  considering the features of both settings combined together. This has led to an extension of the module-checking framework to \ATL/\ATLStar\ specifications~\cite{JamrogaM14,JamrogaM15}. 
Notably, it has been showed that \ATL/\ATLStar\ module checking is strictly more expressive than both \CTL/\CTLStar\ module checking and \ATL/\ATLStar\ model checking~\cite{JamrogaM14,JamrogaM15}.  The computational complexity aspects have been shortly discussed in~\cite{JamrogaM15}, where it is  claimed that the complexity of \ATL/\ATLStar\ module checking is not worse than that of \CTL/\CTLStar\ module checking.

In this paper, we demonstrate that the claim made in~\cite{JamrogaM15} is not correct for \ATLStar. While \ATL\ module checking has the same complexity as \CTL\ module checking, \ATLStar\ module checking turns out to be exponentially harder than \CTLStar\ module checking, and precisely, \THREEXPTIME-complete with a \PTIME-complete complexity for a fixed-size formula\footnote{The incorrect claim in~\cite{JamrogaM15} was due  to a misleading interpretation of the result due to Schewe regarding  \TWOEXPTIME-completeness for the \ATLStar\ satisfiability problem~\cite{Schewe08}.}. The upper bounds are obtained by applying an automata-theoretic approach. 	
The matching lower bound for \ATLStar\ is shown by a technically non-trivial reduction from the word problem for \TWOEXPSPACE-bounded  alternating
Turing Machines. 
\vspace{0.1cm}

{\bf Related work.}
Module checking was introduced in~\cite{KV96}, 
and later extended in several directions.
In \cite{Kupferman97revisited}, the basic \CTL/\CTLStar module-checking problem was extended to the setting where the environment has
imperfect information about the state of the system.
In~\cite{Bozzelli10pushdown}, it was extended to infinite-state open systems by considering pushdown modules.
The pushdown module-checking problem was first investigated for perfect information,
and later, in~\cite{Aminof13pushdown-jc,Bozzelli11newpushown}, for imperfect information;
the latter variant was proved in general undecidable in~\cite{Aminof13pushdown-jc}.
%
\cite{Ferrante08enriched} address module checking against $\mu$-calculus specifications, and in~\cite{Murano08hierarchical}, the module-checking problem was studied for bounded pushdown modules (or \emph{hierarchical modules}).
From a more practical point of view, \cite{Martinelli07specification} built a semi-automated tool for module checking against the existential fragment of \CTL, both in the perfect and imperfect information setting.
A tableaux-based approach to \CTL\ module-checking was also exploited in~\cite{Basu07local}.
Finally, an extension of module checking was used to reason about three-valued abstractions in~\cite{Alfaro04three,Godefroid03open}.\vspace{0.1cm}

%% file: Preliminaries.tex
\section{Preliminaries}\label{sec:Preliminaries}

We fix the following notations. Let $\Prop$ be a finite nonempty set of atomic propositions,
$\Agents$ be a finite nonempty set of agents, and $\Actions$ be a finite nonempty set of actions  that can be made by agents.
For a set  $A\subseteq \Agents$ of agents, an $A$-decision $\decision_A$  is an element in $\Actions^{A}$ assigning to each agent $\agent\in A$ an action $\decision_A(\agent)$. For $A,A'\subseteq \Agents$ with $A\cap A'=\emptyset$, an $A$-decision  $\decision_A$ and $A'$-decision $\decision_{A'}$, $\decision_A\cup \decision_{A'}$ denotes the $(A\cup A')$-decision   defined in the obvious way.
Let $\Decisions = \Actions^{\Agents}$ be the set of \emph{full decisions} of all the agents in $\Agents$. 

Let $\Nat$ be the set of natural numbers. For all  $i,j\in\Nat $, with $i\leq j$, $[i,j]$ denotes the set of natural numbers $h$ such that $i\leq h\leq j$.
 For an infinite word $w$ over an alphabet $\Sigma$ and $i\geq 0$, $w(i)$ denotes the $i^{th}$ letter of $w$ and $w_{\geq i}$ the suffix of $w$ given by $w(i)w(i+1)\ldots$. 

Given a set $\Upsilon$ of directions, an (\emph{infinite}) \emph{$\Upsilon$-tree} $T$ is a   prefix closed subset of $\Upsilon^{*}$ such that
for all $\nu\in T$, $\nu\cdot \gamma\in T$ for some $\gamma\in \Upsilon$. Elements of $T$ are called nodes and $\varepsilon$ is the root of $T$. For $\nu\in T$, the set of children of $\nu$ in $T$ is the set of
nodes of the form $\nu\cdot \gamma$ for some $\gamma\in \Upsilon$.  A infinite path of $T$ is an infinite sequence $\pi$ of nodes such that  $\pi(i+1)$ is a child in $T$ of $\pi(i)$
 for all $i\geq 0$.  For an alphabet $\Sigma$, a $\Sigma$-labeled  $\Upsilon$-tree is a pair $\tpl{T, \Lab}$ consisting of a  $\Upsilon$-tree
 and a labelling $\Lab:T \mapsto \Sigma$  assigning to  each node in $T$ a symbol in $\Sigma$.
We extend the labeling $\Lab$ to infinite paths in the obvious way, i.e. $\Lab(\pi)$ denotes the infinite word over $\Sigma$ given by $\Lab(\pi(0))\Lab(\pi(1))\ldots$.
The labeled tree $\tpl{T, \Lab}$ is \emph{complete} if $T=\Upsilon^{*}$.

\subsection{Concurrent Game Structures}

Concurrent game structures (\CGS)~\cite{AlurHK02} generalize labeled transition systems 
 to a setting with multiple agents (or players). They can be viewed as multi-player games in which players perform concurrent actions, chosen strategically as a function of the history of the game.

\begin{definition}[\CGS]\label{def:CGS} A \CGS\ (over $\Prop$, $\Agents$, and $\Actions$) is a tuple $\GS =\tpl{\States,s_0,\Lab,\Trans}$, where $\States$ is a set of states, $s_0\in \States$ is the initial state, $\Lab: \States \mapsto 2^{\Prop}$ maps each state to a set of atomic propositions, and
$\Trans: \States\times \Decisions \mapsto \States \cup \{\top\}$ is a transition function
 that maps a state and a full decision either to a state or to the special symbol $\top$ ($\top$ is for `undefined') such that for all states $s$, there exists  $\decision\in\Decisions$
 so that $\Trans(s,\decision)\in\States$. The \CGS\ $\GS$ is finite if $\States$ is finite.
 Given a set  $A\subseteq \Agents$ of agents, an $A$-decision $\decision_A$, and a state $s$, we say that  \emph{$\decision_A$ is available at state $s$} if there exists an $(\Agents\setminus A)$-decision $\decision_{\Agents\setminus A}$ such that $\Trans(s,\decision_A\cup \decision_{\Agents\setminus A})\in \States$. We denote by $\Decisions_A(s)$ the nonempty set of $A$-decisions available at state $s$.

For a state $s$ and an agent $\agent$, \emph{state $s$ is controlled by $\agent$} if there is a unique $(\Agents\setminus\{a\})$-decision available at state $s$.  Agent \emph{$\agent$ is passive in  $s$} if there is a unique $\{a\}$-decision available at state $s$. A \emph{multi-agent turn-based game} is a \CGS\ where each state is controlled by an agent.
\end{definition}

We now recall the notion of strategy and counter strategy in a \CGS\ $\GS =\tpl{\States,s_0,\Lab,\Trans}$.
For a state $s$, the set of successors of $s$ is the set of states $s'$ such that $s'= \Trans(s,\decision)$ for some full decision $\decision$.
 A \emph{play} is an infinite sequence of states $s_1 s_2 \ldots $
 such that $s_{i+1}$ is a successor of $s_i$ for all $i\geq 1$. A \emph{path} (or \emph{track}) $\nu$ is a nonempty prefix of some play.  Let $\Trk$  be the set of paths in $\GS$.
  Given a set  $A\subseteq \Agents$ of agents, a \emph{strategy for $A$} is a mapping $f_A: \Trk \mapsto \Actions^{A}$ assigning to each path $\nu$  an $A$-decision available at the last state, denoted $\lst(\nu)$, of $\nu$. For a state $s$, the set $\outcome(s,f_A)$ of \emph{plays consistent with $f_A$ starting from  state $s$} is given by 
 $
 \{s_1 s_2 \ldots\mid s_1=s \text{ and }\forall i\geq 1\,\exists d\in \Actions^{\Agents\setminus A}.\, s_{i+1}=\Trans(s_i,f_A(s_1\ldots s_i)\cup d) \}
 $.

  A \emph{counter strategy} $f^{c}_A$ for $A$ is a mapping assigning to each track $\nu$ a function
$f^{c}_A(\nu): \Decisions_A(\lst(\nu))\mapsto \Actions^{\Agents\setminus A}$,
 where the latter  assigns to each $A$-decision $\decision_A$ available at $\lst(\nu)$ an $(\Agents\setminus A)$-decision $\decision_{\Agents\setminus A}$ such that $\Trans(\lst(\nu),\decision_A\cup \decision_{\Agents\setminus A})\in\States$.
For a state $s$, the set  $\outcome(s,f^{c}_A)$   of plays consistent with the counter strategy $f^{c}_A$ starting from state $s$ is given by:
 \[
\{s_1 s_2 \ldots\mid s_1=s \text{ and }\forall i\geq 1\,\exists d\in \Decisions_A(s_i).\, s_{i+1}=\Trans(s_i,d\cup [f^{c}_A(s_1\ldots s_i)](d)) \}
\]
\begin{definition} For a set $\Upsilon$ of directions, a \emph{Concurrent Game $\Upsilon$-Tree} ($\Upsilon$-\CGT) is a \CGS\ $\tpl{T,\varepsilon,\Lab,\Trans}$, where $\tpl{T,\Lab}$ is a $2^{\Prop}$-labeled $\Upsilon$-tree, and for each node $x\in T$, the set of successors of $x$ corresponds to the set of children of $x$ in $T$. Every $\CGS$ $\GS =\tpl{\States,s_0,\Lab,\Trans}$ induces  a
$\States$-\CGT\ 
 $\Unw(\GS)$ obtained by unwinding $\GS$ from the initial state. 
  Formally, $\Unw(\GS)= \tpl{T,\varepsilon,\Lab',\Trans'}$, where  $T$ is the set of elements $\nu$ in $S^{*}$ such that $s_0\cdot \nu$ is a track of $\GS$, and for all $\nu\in T$ and   $\decision \in \Decisions$, $\Lab'(\nu) = \Lab(\lst(\nu))$ and $\Trans'(\nu,\decision) = \Trans(\lst(\nu),\decision)$, where $\lst(\varepsilon)=s_0$.
\end{definition}

\subsection{Alternating-Time Temporal Logics  $\ATLStar$ and $\ATL$}\label{sec:LogicsATL}

We recall the alternating-temporal logics $\ATLStar$ and $\ATL$ proposed by Alur et al.~\cite{AlurHK02} as extensions of the standard branching-time temporal logics $\CTLStar$ and
$\CTL$~\cite{EmersonH86}, where the path quantifiers are replaced by more general parameterized quantifiers  which allow for reasoning about the strategic capability of groups of agents.
For the given sets $\Prop$ and $\Agents$ of atomic propositions and agents, $\ATLStar$ formulas $\varphi$ are defined by the following grammar:
\[
\varphi ::=  \true \ | \ \ p \ | \ \neg \varphi   \ | \ \varphi \vee \varphi \ | \ \Next \varphi\ | \ \varphi \,\until\, \varphi\ | \  \Exists{A}  \varphi
\]
where $p\in \Prop$, $A\subseteq \Agents$, $\Next$ and  $\until$ are the standard
``next'' and ``until'' temporal modalities,   and $\Exists{A}$ is the
 ``existential strategic quantifier" parameterized by a set of agents. Formula $\Exists{A}\varphi$ expresses that the group of agents $A$ has a collective strategy
    to enforce  property $\varphi$. We  use some shorthands:  the universal strategic quantifier   $\Forall{A}\varphi:=\neg\Exists{A}\neg\varphi$, expressing
    that no strategy of $A$ can prevent property $\varphi$,  the eventually  temporal modality $\Eventually \varphi := \true\,\until\, \varphi$, and the always temporal modality
    $\Always \varphi:= \neg\Eventually\neg\varphi$. A \emph{state formula} is a formula where each temporal modality is in the scope of a strategic quantifier.
    A \emph{basic formula} is a state formula of the form  $\Exists{A}  \varphi$.
 The logic $\ATL$ is the  fragment of $\ATLStar$ where each temporal modality is immediately preceded by a strategic quantifier.
  Note that \CTLStar\ (resp., \CTL) corresponds
 to the fragment of $\ATLStar$ (resp., $\ATL$), where only the strategic modalities
 $\Exists{\Agents}$  and $\Exists{\emptyset}$ (equivalent to  the existential and universal path quantifiers  $\E$ and $\A$, respectively) are allowed.

Given a \CGS\ $\GS$ with labeling $\Lab$ and a play $\pi$ of $\GS$, the
satisfaction relation $\GS,\pi  \models \varphi$ for
$\ATLStar$  is defined as follows (Boolean connectives are treated as usual):
\[ \begin{array}{ll}
\GS,\pi  \models p  &  \Leftrightarrow  p \in \Lab(\pi(0)),\\
\GS,\pi  \models \Next \varphi  & \Leftrightarrow   \GS,\pi_{\geq 1} \models \varphi ,\\
\GS,\pi  \models \varphi_1\,\until\, \varphi_2  &
  \Leftrightarrow  \exists\,j\geq 0: \GS,\pi_{\geq j}
  \models \varphi_2
  \text{ and }  \GS,\pi_{\geq k} \models  \varphi_1 \text{ for all }k\in[0,j-1]\\
\GS,\pi \models \Exists{A} \varphi  & \Leftrightarrow \text{for some strategy } f_A \text{ for }A,\,    \GS,\pi' \models \varphi \text{ for all }\pi'\in \outcome(\pi(0),f_A).
\end{array} \]
For a state $s$ of $\GS$, $\GS,s  \models \varphi$ if there is a play $\pi$ starting from $s$ such that
$\GS,\pi  \models \varphi$. Note that if $\varphi$ is a state formula, then for all plays $\pi$ and $\pi'$ from $s$,
$\GS,\pi  \models \varphi$  iff $\GS,\pi'  \models \varphi$.
$\GS$ is a model of $\varphi$, denoted $\GS\models \varphi$, if for the initial state $s_0$,
$\GS,s_0  \models \varphi$. Note that $\GS\models \varphi$ iff $\Unw(\GS)\models \varphi$.  

\begin{remark}
By \cite{Schewe08}, for a state formula of the form $\Forall{A}\varphi$,
$\GS,s\models \Forall{A}\varphi$ iff there is a counter strategy $f^{c}_A$ for $A$ such that for all $\pi\in \outcome(s,f^{c}_A)$,
$\GS,\pi\models \varphi$.
\end{remark}

\subsection{$\ATLStar$ and $\ATL$ Module checking}

Module checking was proposed in~\cite{KV96} for the verification of finite open systems,
that is systems that interact with an environment whose behavior cannot be determined in advance.
In such a framework, the system is modeled by a \emph{module} corresponding to a two-player turn-based game
between the system and the environment. Thus, in a module, the set of states is partitioned into a set of system
states (controlled by the system) and a set of environment states (controlled by the environment).
The module-checking problem takes  two inputs: a module $M$ and a branching-time temporal formula $\psi$. The idea is
that the open system should satisfy the specification $\psi$ no matter how the environment behaves.
Let us consider the unwinding   $\Unw(M)$ of $M$ into an infinite tree. Checking whether $\Unw(M)$ satisfies
$\psi$ is the usual model-checking problem. On the other hand, for an open system,
$\Unw(M)$ describes the interaction of the system with a maximal environment, i.e. an environment
that enables all the external nondeterministic choices. In order to take into account all the
possible behaviors of the environment, we have to consider all the trees $T$ obtained from $\Unw(M)$
by pruning subtrees whose root is a successor of an environment state (pruning these subtrees
correspond to disabling possible environment choices). Therefore, a module $M$ satisfies
$\psi$ if all these trees $T$ satisfy $\psi$. 
It has been recently proved~\cite{JamrogaM14}  that module checking of \CTL/\CTLStar\ includes two features inherently absent in the semantics of \ATL/\ATLStar, namely irrevocability of strategies and nondeterminism of strategies. On the other hand, 
temporal logics like \CTL\ and \CTLStar\ do not accommodate strategic reasoning.  These facts have  motivated the extension of module checking to a multi-agent setting for handling specifications in \ATLStar\ \cite{JamrogaM15}.  We now recall this setting   which turns out to be more expressive than both $\CTLStar$ module checking and $\ATLStar$ model checking~\cite{JamrogaM14,JamrogaM15}. In this framework, one considers a generalization of modules, namely \emph{open \CGS} (called multi-agent modules in~\cite{JamrogaM15}).

\begin{definition}[Open \CGS] An open \CGS\  is a $\CGS$ $\GS =\tpl{\States,s_0,\Lab,\Trans}$ containing a special agent called ``the environment" ($\env\in\Agents$). Moreover,
for every state $s$, either $s$ is controlled by the environment (environment state) or the environment is passive in $s$ (system state).
\end{definition}

For an open $\CGS$ $\GS=\tpl{\States,s_0,\Lab,\Trans}$, the set of \emph{$($environment$)$   strategy trees of $\GS$}, denoted $\exec(\GS)$, is the set of $\States$-\CGT\ obtained from $\Unw(\GS)$ by possibly pruning some environment transitions.  Formally, $\exec(\GS)$ is the set of $\States$-\CGT\ $\Tau =\tpl{T,\varepsilon,\Lab',\Trans'}$ such that $T$ is a prefix closed subset of the set of $\Unw(\GS)$-nodes  and for all $\nu\in T$  and   $\decision \in \Decisions$, $\Lab'(\nu) = \Lab(\lst(\nu))$, and $\Trans'(\nu,\decision) = \Trans(\lst(\nu),\decision)$ if $\nu\cdot \Trans(\lst(\nu),\decision)\in T$, and $\Trans(\lst(\nu),\decision)=\top$ otherwise,  where $\lst(\varepsilon)=s_0$. Moreover, for all $\nu \in T$, the following holds:
\begin{compactitem}
\item if $\lst(\nu)$ is a system state, then for each successor $s$ of $\lst(\nu)$ in $\GS$, $\nu\cdot s \in T$;
\item if $\lst(\nu)$ is an environment state, then there is a nonempty subset $\{s_1,\ldots,s_n\}$ of the set of $\lst(\nu)$-successors such that the set of children of $\nu$ in
$T$ is $\{\nu\cdot s_1,\ldots,\nu\cdot s_n\}$.
\end{compactitem}\vspace{0.2cm}

Intuitively, when $\GS$ is in a system state $s$, then all the transitions from $s$ are enabled. When
 $\GS$ is instead in an environment state, the set of enabled transitions from $s$ depend on the current environment.
 Since the behavior of the environment is nondeterministic, we have to consider
  all the possible subsets of the set of $s$-successors.
The only constraint, since we consider environments that cannot block the system, is that
not all the transitions from $s$ can be disabled. 
For an open $\CGS$ $\GS$  and an \ATLStar\ formula $\varphi$, $\GS$ \emph{reactively satisfies} $\varphi$, denoted $\GS\models^{r} \varphi$, if
for all strategy trees $\Tau\in \exec(\GS)$, $\Tau\models \varphi$.
Note that $\GS\models^{r} \varphi$ implies $\GS\models \varphi$ (since
$\Unw(\GS)\in\exec(\GS)$), but the converse in general does not hold. 
The \emph{$($finite$)$ module-checking problem against} \ATL\ (resp., \ATLStar) is checking for a given
finite open \CGS\ $\GS$ and an \ATL\ formula (resp., \ATLStar\ state formula)  $\varphi$ whether $\GS\models^{r}\varphi$.

%% file: DecisionProcedures.tex
\section{Decision procedures}\label{sec:DecisionProcedures}

In this section, we provide an automata-theoretic framework
for solving the module-checking problem against \ATL\ and \ATLStar, which is based on the  use of
parity alternating automata for \CGS\ (parity \ACG)~\cite{ScheweF06}.
The proposed approach consists of two steps. For a finite $\CGS$ $\GS$ and an $\ATL$ formula (resp., $\ATLStar$ state formula)  $\varphi$,
one first builds a parity \ACG\ $\Au_{\neg\varphi}$ accepting the set of $\CGT$ which satisfy $\neg\varphi$.
Then $\GS\models_r\varphi$ iff \emph{no} strategy tree of $\GS$ is accepted by $\Au_{\neg\varphi}$.

The rest of the section is organized as follows. In Subsection~\ref{sec:translationLogics}, we recall the framework of \ACG\ and provide a translation of $\ATLStar$ state formulas into equivalent parity \ACG\ involving a double exponential blowup. For \ATL, a linear-time translation into equivalent parity \ACG\ of index $2$
directly follows from~\cite{ScheweF06}. Then, in Subsection~\ref{sec:UpperBoundsATL}, we show that given a finite \CGS\ $\GS$ and a parity $\ACG$ $\Au$, checking that no strategy tree of
$\GS$ is accepted by $\ACG$ can be done in time singly exponential in the size of $\Au$ and polynomial in the size of $\GS$.

\subsection{From  \ATLStar\ to parity \ACG}\label{sec:translationLogics}

First, we recall the class of parity \ACG\ \cite{ScheweF06}.
For a  set $X$,
$\B^{+}(X)$   denotes the set
of \emph{positive} Boolean formulas over $X$, i.e. Boolean formulas  built from elements in $X$
using $\vee$ and $\wedge$. 

A parity \ACG\ over $2^{\Prop}$ and \Agents\ is a tuple
 $\Au=\tpl{Q,q_0,\delta,\alpha}$, where
 $Q$ is a finite set of states, $q_0\in Q$ is the initial state,
  $\delta:Q\times 2^{\Prop}\rightarrow \B^{+}(Q\times \{\Box,\Diamond\}\times 2^{\Agents})$ is the transition function, and $\alpha: Q \mapsto \Nat$ is a parity acceptance condition over $Q$ assigning to each state a color.
The transition function $\delta$ maps a state and an input letter to a positive Boolean combination of universal atoms $(q,\Box,A)$ which refer to \emph{all} successors states for some available $A$-decision, and existential atoms $(q,\Diamond,A)$ which refer to \emph{some} successor state for all available $A$-decisions.
The \emph{index} of $\Au$ is the number of colors in $\alpha$, i.e., the cardinality of $\alpha(Q)$. The size $|\Au|$ of $\Au$ is $|Q|+ |\Atoms(\Au)|$, where
$\Atoms(\Au)$ is the set of atoms  of $\Au$, i.e. the set of tuples in $Q\times \{\Box,\Diamond\}\times 2^{\Agents}$ occurring in the transition function $\delta$ of $\Au$.

 We interpret the parity \ACG\ $\Au$ over $\CGT$. Given a $\CGT$ $\Tau =\tpl{T,\varepsilon,\Lab,\Trans}$ over $\Prop$ and $\Agents$,
 a run of $\Au$ over $\Tau$ is a   $(Q\times T)$-labeled  $\Nat$-tree
$r=\tpl{T_r,\Lab_r}$, where each node  of $T_r$ labelled by $(q,\nu)$ describes a copy of the automaton that is in the state $q$ and reads the node $\nu$ of
$T$. Moreover, we require that $r(\varepsilon)=(q_0,\varepsilon)$ (initially, the automaton is in state $q_0$ reading the root node), and for each  $y\in T_r$ with $r(y)=(q,\nu)$, there is a set $H\subseteq Q\times \{\Box,\Diamond\}\times 2^{\Agents}$ such that $H$ is  model of
$\delta(q,\Lab(\nu))$ and the set $L$ of labels associated with the  children of $y$ in $T_r$ minimally satisfies the following conditions:
\begin{itemize}
  \item for all universal atoms $(q',\Box,A)\in H$, there is an available $A$-decision $d_A$ in the node $\nu$ of $\Tau$ such that for all the children $\nu'$ of $\nu$
  which are consistent with  $d_A$, $(q',\nu')\in L$;
  \item for all existential atoms $(q',\Diamond,A)\in H$ and for all available $A$-decisions $d_A$ in the node $\nu$ of $\Tau$, there is some child $\nu'$ of $\nu$
  which is consistent with $d_A$ such that $(q',\nu')\in L$.
\end{itemize}

The run $r$ is accepting if for all infinite paths $\pi$ starting from the root, the highest color of the states appearing infinitely often along $\Lab_r(\pi)$ is even. The language $\Lang(\Au)$ accepted by $\Au$ consists of the $\CGT$ $\Tau$ over $\Prop$ and $\Agents$ such that there is an accepting run of $\Au$
over $\Tau$.

It is well-known that \ATLStar\ satisfiability has the same complexity as \CTLStar\ satisfiability, i.e., it is \TWOEXPTIME-complete~\cite{Schewe08}. In particular,
given an \ATLStar\ state formula $\varphi$, one can construct in singly exponential time a parity \ACG\ accepting the set of \CGT\ satisfying some special requirements (depending on $\varphi$)
which provide a necessary and sufficient condition for ensuring the existence of some model of $\varphi$~\cite{Schewe08}. These requirements are based on an equivalent representation
 of the models of a formula obtained by a sort of widening operation. When applied to the strategy trees of a finite \CGS, such an encoding is not regular
 since one has to require that for all nodes in the encoding which are copies of the same environment node in the given strategy tree, the associated subtrees are isomorphic.
 Hence, the approach exploited in~\cite{Schewe08} cannot be applied to the module-checking setting. Here, by adapting the construction in~\cite{Schewe08}, we provide a double exponential-time translation of \ATLStar\ state formulas  into equivalent parity \ACG.
In particular, we establish the following result, where for a finite set $B$ disjunct from $\Prop$ and a  $\CGT$ $\Tau =\tpl{T,\varepsilon,\Lab,\Trans}$ over $\Prop$, a
 $B$-labeling extension of $\Tau$ is a \CGT\ over $\Prop\cup B$ of the form  $\tpl{T,\varepsilon,\Lab',\Trans}$, where $\Lab'(\nu)\cap \Prop=\Lab(\nu)$ for all $\nu\in T$.

 \begin{theorem}\label{theo:TranslationATLStar} For an \ATLStar\ state formula $\Phi$ over $\Prop$, one can construct in doubly exponential time a parity  $\ACG$ $\Au_\Phi$  over $2^{\Prop\cup B_\Phi}$, where $B_\Phi$ is the set of basic subformulas of $\Phi$, such that for all
$\CGT$ $\Tau$ over $\Prop$, $\Tau$ is a model of $\Phi$ iff there exists a $B_\Phi$-labeling extension of $\Tau$ which is accepted by
 $\Au_\Phi$.  Moreover, $\Au_\Phi$ has
size $O( 2^{2^{O(|\Phi|\cdot \log(|\Phi|))}})$ and  index $2^{O(|\Phi|)}$.
\end{theorem}

We now illustrate the proof of Theorem~\ref{theo:TranslationATLStar}. For an  $\ATLStar$ formula $\varphi$ over $\Prop$, a \emph{first-level basic subformula} of $\varphi$ is a basic subformula of $\varphi$ for which there is an occurrence in $\varphi$ which is not in the scope of any strategy quantifier. Note that an  \ATLStar\ formula $\varphi$ can be seen as a standard \LTL\ formula~\cite{Pnueli77}, denoted $[\varphi]_\LTL$, over the set $\Prop$ augmented with the set of first-level basic subformulas of $\varphi$. In particular, if $\varphi$ is  a state formula, then $[\varphi]_\LTL$ is a propositional formula.
 Fix an $\ATLStar$ state formula $\Phi$ over $\Prop$, and let $B_\Phi$ be the set of basic subformulas of $\Phi$.
Given a basic subformula  $\Exists{A}\psi\in B_\Phi$  and a $\CGT$ $\Tau=\tpl{T,\varepsilon,\Lab_\Phi,\Trans}$ over
$\Prop\cup B_\Phi$, $\Tau$ is \emph{positively} (resp., \emph{negatively})
\emph{well-formed with respect to }$\Exists{A}\psi$ if:
 \begin{itemize}
   \item for all nodes $\nu\in T$ such that $\Exists{A}\psi\in \Lab_\Phi(\nu)$ (resp., $\Exists{A}\psi\notin \Lab_\Phi(\nu)$), there exists a
   strategy $f_A$ (resp., counter strategy $f_A^{c}$) in $\Tau$ for the set $A$ of agents   such that for all plays $\pi$ in $\Tau$
   starting from $\nu$ which are consistent with $f_A$ (resp., $f_A^{c}$), it holds that $\Lab_\Phi(\pi)$ is a model of the \LTL\  formula
   $[\psi]_\LTL$ (resp., $[\neg\psi]_\LTL$).
 \end{itemize}

\noindent The  $\CGT$ $\Tau$ is \emph{well-formed with respect to} $\Phi$ if:
 (i) for all basic subformulas $\Exists{A}\psi\in B_\Phi$, $\Tau$ is both positively and negatively well-formed w.r.t.  $\Exists{A}\psi\in B_\Phi$, and (ii)
 $\Lab_\Phi(\varepsilon)$ is a model of the propositional formula  $[\Phi]_\LTL$.
 The following proposition easily follows from the semantics of $\ATLStar$ and the remark at the end of Section~\ref{sec:LogicsATL}.

\begin{proposition}\label{prop:WelFormedness} Given a $\CGT$ $\Tau$ over $\Prop$, $\Tau$ is a model of
$\Phi$ iff there exists a $B_\Phi$-labeling extension of $\Tau$ which is well-formed w.r.t. $\Phi$.
\end{proposition}

We show the following result that together with Proposition~\ref{prop:WelFormedness} provides a proof of Theorem~\ref{theo:TranslationATLStar}.

\newcounter{theo-FormulaBadFGNested}
\setcounter{theo-FormulaBadFGNested}{\value{theorem}}

 \begin{theorem}\label{theo:ACGforATLStar} Given an $\ATLStar$ state formula $\Phi$, one can construct in time doubly exponential in the size of $\Phi$, a parity  $\ACG$ $\Au_\Phi$  over $2^{\Prop\cup B_\Phi}$
 accepting the set of $\CGT$ over $\Prop\cup B_\Phi$ which are well-formed w.r.t. $\Phi$.   Moreover, $\Au_\Phi$ has
size $O( 2^{2^{O(|\Phi|\cdot \log(|\Phi|))}})$ and  index $2^{O(|\Phi|)}$.
\end{theorem}

In order to prove Theorem~\ref{theo:ACGforATLStar}, we exploit
the well-known translation of \LTL\  into  B\"{u}chi nondeterministic
word automata (B\"{u}chi \NWA)~\cite{VardiW94}.
In particular, given an $\LTL$ formula $\psi$, one can construct in singly exponential time a B\"{u}chi  \NWA\  accepting the set of infinite words which are models of $\psi$~\cite{VardiW94}.
In order to handle a basic subformula of the form $\Exists{\Agents}\psi$ and its negation  ($\Exists{\Agents}$  and $\neg\Exists{\Agents}$ correspond to the existential and universal path quantifiers of $\CTLStar$), it suffices to use the B\"{u}chi \NWA\ $\Au_\psi$ associated with $[\psi]_\LTL$ and the dual
$\tilde{\Au}_\psi$ of $\Au_\psi$, respectively ($\tilde{\Au}_\psi$ is a universal co-B\"{u}chi word automaton).
Indeed, for checking that $\Exists{\Agents}\psi$ holds at the current node $\nu$ of the input, the \ACG\ simply guesses an infinite path $\pi$ from $\nu$ and simulates a run
of $\Au_\psi$ over the labeling of $\pi$, and checks that it is accepting by using its parity acceptance condition.
Similarly, for the formula  $\neg\Exists{\Agents}\psi$, the \ACG\ simulates the  universal co-B\"{u}chi word automaton
$\tilde{\Au}_\psi$ for checking that all the plays starting from $\nu$  satisfy the \LTL\ formula $[\neg\psi]_\LTL$. This reasoning is the key
for translating $\CTLStar$ formulas into equivalent parity alternating tree automata with a single exponential blowup~\cite{KVW00}.
However, for handling more general basic subformulas $\Exists{A}\psi$ and their negations, we need to use deterministic word automata
for the  \LTL\ formulas $[\psi]_\LTL$ and $[\neg\psi]_\LTL$. This because the choices of an $\ACG$ are local, and the set  of
plays starting from the current input node which are consistent with a strategy (resp., counter strategy) of $A$ may be infinite and properly contained in the set of all the plays starting from $\nu$. The determinization of a B\"{u}chi \NWA\ involves an additional exponential blowup~\cite{Safra88}.

\subsection{Upper bounds for \ATL\ and \ATLStar\ module checking}\label{sec:UpperBoundsATL}

In this section, we establish the following result.

 \begin{theorem}\label{theo:EmptinessACG} Given a \CGS\ $\GS$ over $\Prop$, a finite set $B$ disjunct from $\Prop$, and a parity \ACG\ $\Au$ over $2^{\Prop\cup B}$,
  checking whether there are no $B$-labeling extensions of strategy trees of $\GS$ accepted by $\Au$ can be done in time singly exponential in the size of $\Au$ and polynomial in the size of $\GS$.
\end{theorem}

By~\cite{Schewe08}, $\ATL$ can be translated in linear time into equivalent parity \ACG\ of index $2$.
Thus, by Theorem~\ref{theo:TranslationATLStar} and Theorem~\ref{theo:EmptinessACG}, and since the \CTL\ module-checking problem is \EXPTIME-complete, and \PTIME-complete for a fixed \CTL\ formula,  we obtain the following corollary.

\begin{corollary} The $\ATLStar$ module-checking problem is in \THREEXPTIME\ while the $\ATL$ module-checking problem is \EXPTIME-complete. Moreover, for a fixed
$\ATLStar$ state formula (resp., $\ATL$ formula), the module-checking problem is \PTIME-complete.
\end{corollary}

In Section~\ref{sec:LowerBound}, we provide a lower bound for the $\ATLStar$ module-checking problem matching the upper bound in the corollary above. We now illustrate the proof of Theorem~\ref{theo:EmptinessACG}. We assume that the set $B$ in the statement of Theorem~\ref{theo:EmptinessACG} is empty (the general case where $B\neq \emptyset$ is similar).
Let  $\GS =\tpl{\States,s_0,\Lab,\Trans}$ be a finite \CGS\ over $\Prop$. Note that the transition function $\Trans'$ of a strategy tree $\Tau =\tpl{T,\varepsilon,\Lab',\Trans'}$ of $\GS$ is completely determined by $T$ and the transition function $\Trans$ of $\GS$. Hence, for the fixed $\CGS$ $\GS$, $\Tau$ can be simply specified by the underlying $2^{\Prop}$-labeled tree $\tpl{T,\Lab'}$. We consider an equivalent representation of $\tpl{T,\Lab'}$ by the $(2^{\Prop}\cup\{\bot\})$-labeled \emph{complete} $\States$-tree $\tpl{\States^{*},\Lab_\bot}$, called the
$\bot$-completion encoding of $\Tau$ ($\bot$ is a fresh proposition  used to denote ``completion" nodes), defined as: for each concrete node $\nu\in T$, $\Lab_\bot(\nu)=\Lab'(\nu)$, while for each completion node $\nu\in \States^{*}\setminus T$, $\Lab_\bot(\nu)=\{\bot\}$.

By the semantics of $\ACG$, given a parity $\ACG$ $\Au$ with $n$ states and index $k$, we can easily construct in polynomial time a standard parity alternating tree automaton (\ATA) $\Au_\GS$ over the alphabet $\States\times (2^{\Prop}\cup\{\bot\})$ and the set $\States$ of directions, having $O(n)$-states and index $k$, accepting the set
of $\States\times (2^{\Prop}\cup\{\bot\})$-labeled complete $\States$-trees $\tpl{\States^{*},\Lab}$ such that for each $\nu\in T$, the $\States$-label of $\nu$ coincides with the direction $\lst(\nu)$, and the labeled tree obtained from $\tpl{\States^{*},\Lab}$ by removing the $\States$-labeling component is the $\bot$-completion encoding of a strategy tree of  $\GS$ accepted by $\Au$. However, this approach has an inconvenient. Indeed, in order to check emptiness of the parity $\ATA$ $\Au_\GS$, one first construct an equivalent parity nondeterministic tree automaton (\NTA) $\Au'_\GS$, and then check for emptiness of $\Au'_\GS$. By~\cite{EJ88,Var98}, $\Au'_\GS$ has index polynomial in the size of the $\ACG$ $\Au$, and number of states which is singly exponential both in the size of $\Au$ and in the number of directions, which in our case, coincides with the number of $\GS$-states.
We show that due to the form of the transition function of an $\ACG$ (it is independent of the set of directions), the exponential blowup in the number of $\GS$-states can be
avoided. In particular, by adapting the construction provided in~\cite{Var98} for converting parity two-way \ATA\ into equivalent parity \NTA, we provide a direct translation into parity \NTA\ as established in the following Theorem~\ref{Theor:FromACGToNTA}. Since nonemptiness of parity \NTA\ with $n$ states and index $k$ can be solved in time
$O(n^{k})$~\cite{KupfermanV98}, by Theorem~\ref{Theor:FromACGToNTA}, Theorem~\ref{theo:EmptinessACG} (for the case $B=\emptyset$) directly follows.

\newcounter{theo-FromACGToNTA}
\setcounter{theo-FromACGToNTA}{\value{theorem}}

\begin{theorem}\label{Theor:FromACGToNTA} Given a finite $\CGS$ $\GS =\tpl{\States,s_0,\Lab,\Trans}$ over $\Prop$ and an $\ACG$ $\Au=\tpl{Q,q_0,\delta,\alpha}$
over $2^{\Prop}$ with index $k$, one can construct in singly exponential time, a parity $\NTA$ $\Au_\GS$ over $2^{\Prop}\cup\{\bot\}$ and the set $\States$ of directions such that
$\Au_\GS$ accepts the set of $2^{\Prop}\cup \{\bot\}$-labeled complete $\States$-trees which are the $\bot$-completion encodings of the strategy trees of $\GS$
which are accepted by $\Au$. Moreover, $\Au_\GS$ has index $O(k|\Au|^{2})$ and $O(\States \cdot  (k|\Au|^{2})^{O(k|\Au|^{2})})$ states.
\end{theorem}

%% file: LowerBounds.tex
\section{\THREEXPTIME--hardness of \ATLStar\ module checking}\label{sec:LowerBound}

In this section, we establish the following result.

\begin{theorem}\label{theorem:lowerbound}
Module checking against \ATLStar\ is
    \THREEXPTIME--hard even for two-player turn-based open $\CGS$ of fixed size.
\end{theorem}

Theorem~\ref{theorem:lowerbound} is proved by a polynomial-time  reduction from the word problem for \TWOEXPSPACE--bounded alternating
 Turing Machines. 
 Formally,
an alternating  Turing Machine  (TM, for short)  is a tuple
$\mathcal{M}=\tpl{\Sigma,Q,Q_{\forall},Q_{\exists},q_0,\delta,F}$,
where $\Sigma$ is the input alphabet, which contains the blank
symbol $\#$, $Q$ is the finite set of states which is partitioned
into $Q=Q_{\forall} \cup  Q_{\exists}$, $Q_{\exists}$ (resp.,
$Q_{\forall}$) is the set of existential (resp., universal)
states, $q_0$ is the initial state, $F\subseteq Q$ is the set  of
accepting states, and the transition function $\delta$ is a
mapping $\delta: Q\times \Sigma\rightarrow (Q \times\Sigma \times
\{L,R\})^2$. Configurations of $\mathcal{M}$ are words in
$\Sigma^*\cdot(Q\times \Sigma)\cdot\Sigma^*$. A configuration
$C= \eta\cdot(q,\sigma)\cdot\eta'$ denotes that the tape content is
$\eta\cdot\sigma\cdot\eta'$, the current state (resp., input symbol) is $q$ (resp., $\sigma$), and the reading head
is at position $|\eta|+1$. From configuration $C$,  the machine $\mathcal{M}$ 
nondeterministically chooses a triple $(q',\sigma',dir)$ in
$\delta(q,\sigma)=
\tpl{(q_l,\sigma_l,dir_l),(q_r,\sigma_r,dir_r)}$, and then moves
to state $q'$, writes $\sigma'$ in the current tape cell, and its
reading head moves one cell to the left or to the right, according
to $dir$. We denote by $succ_l(C)$ and
$succ_r(C)$ the successors of $C$ obtained by choosing respectively
the left and the right triple in
$\tpl{(q_l,\sigma_l,dir_l),(q_r,\sigma_r,dir_r)}$. The
configuration $C$ is accepting  (resp., universal, resp., existential ) if the associated state $q$
is in $F$ (resp., in $Q_\forall$, resp., in $Q_\exists$). Given an input $\alpha\in \Sigma^*$, a (finite) computation tree
of $\mathcal{M}$ over $\alpha$ is a finite tree in which each node is labeled
by a configuration. The root of the tree corresponds to the
initial configuration associated with $\alpha$. 
 An \emph{internal} node that is labeled by a universal configuration
$C$  has two
children, corresponding to $succ_l(C)$ and $succ_r(C)$, while an internal
node labeled by an existential configuration $C$  has a single child,
corresponding to either $succ_l(C)$ or $succ_r(C)$. The tree is
accepting iff every leaf is labeled by an accepting configuration. An input $\alpha\in\Sigma^*$ is \emph{accepted} by
$\mathcal{M}$ iff there is an accepting computation tree
of $\mathcal{M}$ over $\alpha$.
If $\mathcal{M}$ is \TWOEXPSPACE--bounded, then there
 is a constant
  $k\geq 1$ such that for each  $\alpha\in \Sigma^*$, the
 space needed by $\mathcal{M}$ on input $\alpha$ is bounded by
  $2^{2^{|\alpha|^k}}$. It is well-known~\cite{CKS81} that \THREEXPTIME\ coincides with
 the class of all languages accepted by  \TWOEXPSPACE--bounded
  alternating Turing Machines (TM). Moreover, the considered word problem remains \THREEXPTIME-complete even for
  \TWOEXPSPACE--bounded TM of fixed size.

  Fix
   a  \TWOEXPSPACE--bounded TM
$\mathcal{M}=\tpl{\Sigma,Q,Q_{\forall},Q_{\exists},q_0,\delta,F}$ and an input $\alpha\in \Sigma^*$. Let $n=|\alpha|$. W.l.o.g. we  assume that the constant $k$ is $1$.
 Hence, any reachable configuration of $\mathcal{M}$
over $\alpha$ can be seen as a word in $\Sigma^*\cdot(Q\times
\Sigma)\cdot\Sigma^*$ of length exactly $2^{2^{n}}$. In particular,
 the initial configuration is  $(q_0,\alpha(0)) \alpha(1) \ldots
\alpha(n-1)\cdot (\#)^{2^n-n}$. Note that
for a  TM configuration $C=u_1 u_2\ldots u_{2^{2^{n}}}$ and for all $i\in [1, 2^{2^{n}}]$ and $\dire\in \{l,r\}$, the
value $u'_i$ of the $i$-th cell of $succ_{\dire}(C)$ is completely determined by the values
$u_{i-1}$, $u_{i}$ and $u_{i+1}$ (taking $u_{i+1}$ for $i=2^n$ and
 $u_{i-1}$ for $i=1$ to be some special symbol, say $\bot$). We denote by
 $\Succ_{\dire}(u_{i-1},u_i,u_{i+1})$  our
expectation for $u'_i$ (this function  can be trivially obtained from the transition function of
$\mathcal{M}$).
According to the above observation, we use the set $\Lambda$ of triples of the form $(u_p,u,u_s)$ where $u\in  \Sigma \cup (Q\times \Sigma)$, and $u_p,u_s\in \Sigma \cup (Q\times \Sigma)\cup \{\bot\}$. 

In the following, we prove the following result from which Theorem~\ref{theorem:lowerbound} directly follows.

\begin{theorem}\label{theorem:lowerboundReduction} One can construct, in time polynomial in $n$ and the size of $\mathcal{M}$, a finite turn-based open $\CGS$
$\GS$ and an $\ATLStar$ state formula $\varphi$ over the set of agents $\Agents =\{\sys,\env\}$ such that
$\mathcal{M}$ accepts $\alpha$ \emph{iff}
there is a strategy tree in $ \exec(\GS)$ that satisfies $\varphi$  \emph{iff}
$\GS\not\models_r \neg \varphi$. Moreover, the size of $\GS$ depends only on the size of $\mathcal{M}$.
\end{theorem}

In order to prove Theorem~\ref{theorem:lowerboundReduction}, we first define a suitable encoding of the accepting computation trees of
$\mathcal{M}$ over $\alpha$.\vspace{0.2cm}

\input{FigureLowerBound}
\noindent \textbf{Encoding of   computation trees of
$\mathcal{M}$ over $\alpha$.}
In the encoding of a TM configuration, as usual, for each TM cell, we record both  the content of the cell and  the location (cell number)
of the cell on the TM tape. We also record  the contents of the previous and next cell (if any).  Since the cell number is in the range $[0,2^{2^n}-1]$,
it can be encoded  by a $2^n$-bit counter. Moreover, we  need
an $n$-bit counter in order to keep track of the \emph{position}
(index) of each bit of our $2^n$-bit counter.
Formally, we exploit the following set $\Prop$ of atomic propositions
\[
\Prop:= \Lambda\cup \{0,1,
\forall,\exists,l,r,f,  \begOne,\EndOne,\checkOne,\begTwo,\EndTwo, \checkTwo, \markTwo\}
\]
 where  $0$ and $1$ are used to encode the cell numbers, and the meaning of the letters   in
    $\{\forall,\exists,l,r,f,  \begOne,\EndOne,$ $\checkOne,\begTwo,\EndTwo, \checkTwo, \markTwo\}$
will be explained later.

The value $b\in\{0,1\}$ and the index   $i\in [0,2^{n}-1]$ of a bit in the $2^n$-bit counter is encoded by a  \emph{TM sub-block} $sb$,  which is a word  of the form $sb =  \textit{Type} \cdot\Tag \cdot\{b\}\cdot\{b_1\}\cdot\ldots\cdot \{b_n\}\cdot\{\EndTwo\}$,
where $\textit{Type}\in \{\{\begTwo\},\{\checkTwo\}\}$,    $\Tag\in\{\emptyset,\{\markTwo\}\}$, and $b_1\cdot\ldots\cdot b_n\in\{0,1\}^{n}$ is the binary code of the index $i$.
We say that $b$ (resp., $i$) is the \emph{content} (resp., \emph{number}) of $sb$.
 Moreover, $sb$ is  a \emph{main} (resp., \emph{check}) sub-block if $\textit{beg} =\{\begTwo\}$ (resp., $\textit{beg} =\{\checkTwo\}$), and
  $sb$ is \emph{marked}  (resp., $\emph{non-marked}$) if
$\Tag = \{\markTwo\}$  (resp., $\Tag =\emptyset$).

A TM cell is in turn encoded by a \emph{TM block},  which   is a word $bl$ of the form
$bl=\{\begOne\}\cdot \Tag \cdot \lambda \cdot   sb_1 \cdot \ldots \cdot sb_k \cdot\{\EndOne\}$
for some $k\geq 1$, where $\Tag\in\{\emptyset,\{\checkOne\}\}$,  $\lambda\in \Lambda$ is the \emph{content} of $bl$, and $sb_1,\ldots,sb_k$ are non-marked \emph{main} sub-blocks if $\Tag =\emptyset$ (in this case, $bl$ is a \emph{main} block),
and $sb_1,\ldots,sb_k$ are non-marked \emph{check} sub-blocks otherwise (in this case, $bl$ is a \emph{check} block).
  If $k=2^n$ and  for each $i\in [1,2^n]$, the number of $sb_i$ is $i-1$,  we say that $bl$ is  \emph{well-formed}. In this case, the \emph{number}   of $bl$
is the integer in $[0,2^{2^n}-1]$ whose binary code is given by
$b_1\ldots b_{2^{n}}$, where for all $i\in [1,2^{n}]$, $b_i$ is the content of $sb_i$.   Note that if the content $\lambda$ of $bl$ is of the form $(u_p,u,u_s)$, then $u$ represents the value  of the encoded TM cell, while $u_p$
(resp., $u_s$) represents the value of the previous (resp., next) cell in the TM configuration.

 TM configurations $C=u_1 u_2\ldots u_k$ (note that here we do not require that $k=2^{2^n}$) are then encoded by  words $w_C$  of the form
$w_C = \Tag_1 \cdot bl_1 \cdot \ldots \cdot bl_k \cdot \Tag_2$,
where $\Tag_1\in\{\{l\},\{r\}\}$,
 for each $i\in [1,k]$, $bl_i$ is a \emph{non-marked main TM block} whose content is $(u_{i-1},u_i,u_{i+1})$ (where $u_0=\bot$ and $u_{k+1}=\bot$), $\Tag_2=\{f\}$ if $C$ is accepting,  $\Tag_2=\{\exists\}$ if $C$ is  non-accepting and existential, and $\Tag_2=\forall$ otherwise. The symbols $l$ and $r$ are used to mark a left and a right TM successor, respectively. We also use the symbol $l$ to mark the initial configuration.
If $k=2^{2^n}$ and for each $i\in [1,k]$, $bl_i$ is a well-formed  block  having number
$i-1$, then we say that $w_C$ is a \emph{well-formed code} of $C$.
A sequence $w_{C_1}\cdot \ldots \cdot w_{C_p}$ of well-formed TM configuration codes   is \emph{faithful to the evolution} of $\mathcal{M}$ if for each $1\leq i<p $, either $w_{C_{i+1}}$ is marked by symbol $l$ and $C_{i+1}=succ_l(C_i)$, or $w_{C_{i+1}}$ is marked by symbol $r$ and $C_{i+1}=succ_r(C_i)$.

In the encoding of the computation trees of
$\mathcal{M}$, \emph{marked} sub-blocks are used as additional branches for ensuring by a \CTLStar\ formula that the TM blocks are well-formed (i.e., the $n$-counter is properly updated) and the TM configurations codes are well-formed as well (i.e., the $2^{n}$-counter is properly updated). Moreover,
 suitable  tree encodings of   \emph{check} TM blocks, called \emph{block check-trees} (see Figure~\ref{FigureLowerBound}(c)) are
  exploited as additional subtrees   for ensuring by an  \ATLStar\ formula  that the encoding is faithful to the evolution of $\mathcal{M}$.
 Intuitively, a block check-tree corresponds to a check TM block $bl$ extended with additional branches which represent marked copies of the  sub-blocks of $bl$. 

\begin{definition}[Block Check-trees] A \emph{block check-tree} is a $2^{\Prop}$-labeled tree $\tpl{T,\Lab}$ such that there is an infinite path $\pi$ from the root so that
$\Lab(\pi)$ is of the form $bl \cdot \emptyset^{\omega}$, where $bl$ is a check block ($bl$ is the block encoded by $\tpl{T,\Lab}$), and the following holds:
\begin{compactitem}
  \item  each node $x$ of $\pi$  labeled by $\{\checkTwo\}$ (the first symbol of a sub-block of $bl$) has two children, and for the child $y$ of $x$ which is not visited by $\pi$, there is a unique infinite path $\pi'$   from $x$ and visiting $y$. Moreover, $\Lab(\pi')$ is of the form $sb \cdot \emptyset^{\omega}$, where $sb$ is a marked check sub-block ($sb$ is the \emph{companion}  of the main sub-block of $\pi$ associated with node $x$);
   \item  each node of $\pi$ which is not labeled by $\{\checkTwo\}$  has exactly one child.
\end{compactitem}\vspace{0.1cm}
 $\tpl{T,\Lab}$ is  \emph{well-formed} if, additionally, $\Lab(\pi)$ encodes a well-formed check block and for each sub-block $sb$ along $\pi$, the companion $sb'$ of $sb$ has the same content and number as $sb$.
\end{definition}

We now define an encoding of the computation trees of $\mathcal{M}$ (see Figure~\ref{FigureLowerBound}), where, intuitively, the computations paths (main paths) are extended with additional branches (marked main sub-blocks) and additional subtrees (block check-trees).

\begin{definition}[Tree-Codes] A \emph{tree-code} is a  $2^{AP}$-labeled tree $\tpl{T,\Lab}$ such that there is a set $\Pi$ of infinite paths from the \emph{root},
 called \emph{main paths}, so that for each  $\pi\in \Pi$, $\Lab(\pi) =w_\pi\cdot \emptyset^{\omega}$ where
$w_\pi$ is a sequence of codes of
  TM configurations $C_1,\ldots,C_p$, $C_1$ has the form $(q_0,\alpha(0))\alpha(1)\ldots \alpha(n-1)\cdot (\#)^{k}$ for some $k\geq 0$,  $C_p$ is accepting, $C_i$ is not accepting for all $i\in [1,p-1]$, and the following holds for each node $x$ along $\pi$:
\begin{compactitem}
  \item  if $x$ has label $\{\forall\}$, then $x$  has two children, with labels  $\{l\}$ and  $\{r\}$, respectively, and for the child $y$ of $x$ which is not visited by $\pi$, there is a main path visiting $y$;
   \item  if $x$ has label $\{\begTwo\}$, then $x$ has two children, and for the child $y$ of $x$ which is not visited by $\pi$, there is a unique infinite path $\pi'$ starting from $x$ and visiting $y$. Moreover, $\Lab(\pi')$ is of the form $sb \cdot \emptyset^{\omega}$, where $sb$ is a marked main sub-block ($sb$ is the \emph{companion}  of the non-marked main sub-block along $\pi$ associated with node $x$);
      \item  if $x$ has label $\{\begOne\}$, then $x$ has two children, and if we remove the child of $x$ visited by $\pi$ and all its descendants, then the resulting subtree rooted at node $x$ is a \emph{block check-tree};
  \item  if the label of $x$ is not in $\{\{\forall\},\{\begOne\},\{\begTwo\}\}$, then $x$  has exactly one child.
   \end{compactitem}\vspace{0.2cm}
 A tree-code   $\tpl{T,\Lab}$ is \emph{well-formed} if for each main path $\pi$, the following additionally holds:
 \begin{compactitem}
  \item (i)  TM configuration codes along $w_\pi$ are \emph{well-formed}, (ii)  for each sub-block $sb$  along $\pi$, the companion  of $sb$   has the same content and number as $sb$, and (iii)
   for each   block $bl$ along $\pi$, the associated block check-tree is well-formed and encodes a check block having the same number and content as $bl$.
\end{compactitem}
A tree-code is \emph{fair}, if for each main path $\pi$, $w_\pi$ is faithful to the evolution of $\mathcal{M}$. Evidently, there is a fair well-formed tree-code iff there is an accepting computation tree of $\mathcal{M}$ over $\alpha$.
\end{definition}

\noindent  \textbf{Construction of the open $\CGS$ $\GS$ and the $\ATLStar$ formula $\varphi$ in Theorem~\ref{theorem:lowerboundReduction}.}
By the definition of tree-codes,  the following result (Lemma~\ref{lemma:constructionOfCGS}), concerning the construction of the open \CGS\ in Theorem~\ref{theorem:lowerboundReduction},  trivially follows,
where   a \emph{minimal} $2^{AP}$-labeled tree is a $2^{AP}$-labeled tree $\tpl{T,\Lab}$ whose root has label  $\{l\}$ and satisfying the following:
\begin{compactitem}
\item (i) for each node $x$, the children of $x$ have distinct labels and $\Lab(x)$ is either empty or a singleton; (ii)
 each  node labeled by $\{\begTwo\}$ (resp., $\{\begOne\}$) has two children, one with empty label and the other one with label $\{\markTwo\}$ (resp., $\{\checkOne\}$); and
 (iii)
  each node labeled by $\{\forall\}$ has two children, with labels  $\{l\}$ and $\{r\}$, respectively.
\end{compactitem}

\newcounter{lemma-constructionOfCGS}
\setcounter{lemma-constructionOfCGS}{\value{lemma}}

\begin{lemma}\label{lemma:constructionOfCGS} One can construct in time polynomial in $|\Prop|$, a finite turn-based
open $\CGS$ $\GS$ over $\Prop$ and $\Agents =\{\env,\sys\}$ satisfying the following:
\begin{compactitem}
  \item $\Unw(\GS) =\tpl{T,\Lab,\Trans}$, where $\tpl{T,\Lab}$ is a \emph{minimal} $2^{AP}$-labeled tree;
  \item for each \emph{tree-code} $\tpl{T',\Lab'}$, there is a strategy tree in  $\exec(\GS)$ of the form $\tpl{T',\Lab',\Trans'}$;
  \item each state which is labeled by either $\{\begOne\}$ or $\{\begTwo\}$ or $\{\forall\}$ is controlled by the system;
     \item   each state whose label is \emph{not} in $\{\{\begOne\},\{\begTwo\}, \{\forall\}\}$ is controlled by the environment.
\end{compactitem}
\end{lemma}

According to Lemma~\ref{lemma:constructionOfCGS}, a minimal $2^{AP}$-labeled tree can be interpreted as a two-player turn-based \CGT\ between the environment and the system,
where the nodes having label in   $\{\{\begOne\},\{\begOne\},\{\forall\}\}$   are controlled by the system, while all the other nodes are controlled by the environment. With this interpretation, we now establish the following result  that together
with Lemma~\ref{lemma:constructionOfCGS} provide a proof of Theorem~\ref{theorem:lowerboundReduction}.

\newcounter{lemma-constructionOfATLStarFOrmula}
\setcounter{lemma-constructionOfATLStarFOrmula}{\value{lemma}}

\begin{lemma}\label{lemma:constructionOfATLStarFOrmula} One can construct in time polynomial in $n$ and $|\Prop|$, an $\ATLStar$ state formula $\varphi$ over $\Prop$ and
 $\Agents =\{\env,\sys\}$
such that for each minimal $2^{AP}$-labeled tree $\tpl{T,\Lab}$, $\tpl{T,\Lab}$ is a model of $\varphi$ iff $\tpl{T,\Lab}$ is a fair well-formed tree-code.
\end{lemma}
\begin{proof}
The \ATLStar\ formula $\varphi$ is given by $\varphi:= \varphi_\TC \wedge \varphi_\WTC \wedge \varphi_\fair$,
where: (i) $\varphi_\TC$ is a \CTLStar\ formula which is satisfied by a  minimal  $2^{AP}$-labeled tree  $\tpl{T,\Lab}$ iff $\tpl{T,\Lab}$ is a tree-code,
(ii) $\varphi_\WTC$ is a \CTLStar\ formula  requiring that each tree-code   is well-formed, and (iii)
$\varphi_\fair$ is an \ATLStar\ formula ensuring that a well-formed tree-code is fair. Here, we focus on the construction of the \ATLStar\ formula  $\varphi_\fair$.
Let $\tpl{T,\Lab}$ be a well-formed tree-code,
 $\pi$ be a main path of $\tpl{T,\Lab}$, and $w_C$ be a non-terminal well-formed configuration code  along $\pi$ associated with a TM configuration $C$.
Assume that the last symbol of $w_C$ is $\forall$, i.e., $C$ is   universal (the other case, where the last symbol is $\exists$ being similar).  Let $x$ be the node associated with the last symbol of $w_C$.
Then, there are two configuration codes $w_{C_l}$ and $w_{C_r}$ associated with configurations $C_l$ and $C_r$, respectively, such that
the first symbol of $w_{C_l}$ (resp., $w_{C_r}$) is $\{l\}$ (resp., $\{r\}$). Moreover, one of the codes follows $w_C$ along $\pi$, while the other one follows
$w_C$ along a main path which visits the child of node $x$ which is not visited by $\pi$. We have to require that for all $\dir\in \{l,r\}$,
$C_\dir= succ_\dir(C)$. This reduces to check that for each block $bl$ of $w_C$, denoted by $bl_\dir$ the block of
$w_{C_\dir}$ having the  same number as $bl$, and by $(u_p,u,u_s)$ (resp., $(u'_p,u',u'_s)$) the content of block
  $bl$ (resp., $bl_\dir$), the following holds: $u'= \Succ_\dire(u_p,u,u_s)$. 
For this check, we exploit the block check-tree, say $\BCT$,   associated
with the main block $bl_\dir$, whose encoded check TM block (the companion of $bl_\dir$) has the same content and number as $bl_\dir$.
Recall that all the nodes in $\BCT$ but the root (which is a $\{\begOne\}$-labeled node) are controlled by the environment. Moreover, the unique nodes in
$\tpl{T,\Lab}$ controlled by the system are the ones having label in $\{\{\forall\},\{\begOne\},\{\begTwo\}\}$. Let $x_{bl}$ be the starting node for the selected block $bl$
of $w_C$. Then, there is a strategy $f_{bl}$ of the player system such that
 \begin{compactitem}
  \item (i) each play consistent with the strategy $f_{bl}$ starting from node $x_{bl}$ gets trapped in the check-tree $\BCT$, and (ii)
 each infinite path starting from node $x_{bl}$ and leading to some marked sub-block of $\BCT$ is consistent with the strategy
  $f_{bl}$.
\end{compactitem}
Note that each strategy of the system selects exactly one child for each node controlled by the system. Thus, the $\ATLStar$ formula $\varphi_\fair$ ``guesses'' the strategy $f_{bl}$ and ensures that the guess is correct by verifying the following conditions
on the outcomes of $f_{bl}$ from   node $x_{bl}$:
 \begin{compactenum}
  \item  each outcome visits a $\{\checkOne\}$-node whose parent belongs to a  block of $w_{C_\dir}$. This ensures that all the outcomes get trapped
  in the \emph{same} block check-tree associated with some block of $w_{C_\dir}$. Moreover, for the label $(u'_p,u',u'_s)$ of the node following the $\{\checkOne\}$-node
  along the outcome,  $u'= \Succ_\dire(u_p,u,u_s)$, where $(u_p,u,u_s)$ is the content of
$bl$.
  \item  for each outcome $\pi'$ which leads to a marked sub-block $sb'$ (note that this sub-block is necessarily in $\BCT$), denoting by $sb$ the sub-block
  of $bl$ having the same number as $sb$, it holds that $sb$ and $sb'$ have the same content.
\end{compactenum}
The first (resp., second) condition is implemented by the \LTL\ formula $\psi_\dir$ (resp., $\psi_{\textit{cor}}$) in the definition of $\varphi_\fair$ below.
\[
\varphi_\fair  :=    \displaystyle{\bigwedge_{\dire\in\{l,r\}}} \A\Always \Bigl( (\begOne \wedge \E\Eventually \,\dir) \, \longrightarrow \,
 \Exists{\sys}\bigl( \psi_{\dire} \wedge  \psi_{\textit{cor}}  \bigr)\,\,\Bigr)
\]
\vspace{-0.2cm}
\[
\begin{array}{ll}
\psi_{\dire} := & \displaystyle{\bigvee_{(u_p,u,u_s),(u'_p,u',u'_s)\in\Lambda:\, u'= \Succ_\dire(u_p,u,u_s) }}\Bigl(\Next^{2}(u_p,u,u_s) \,\,\wedge\,\, \\
& \quad  \Bigl[(\neg l\wedge \neg r)\,\until\, \bigl(\dire \wedge \Next( (\neg l\wedge \neg r)\, \until\, (\checkOne \wedge \Next (u'_p,u',u'_s)))\bigr)\Bigr]\Bigr)
\end{array}
\]
\vspace{-0.2cm}
\[
\psi_{\textit{cor}} := \Eventually \markTwo \, \rightarrow \, \Bigl(\bigl(\neg\EndOne \wedge (\begTwo \rightarrow \Next\theta_{\textit{cor}})\bigr )\,\until \,\EndOne\Bigr)
\]
\vspace{-0.2cm}
\[
\theta_{\textit{cor}} :=  \Bigl(\displaystyle{\bigwedge_{i=1}^{i=n}\bigvee_{b\in \{0,1\}}}((\Next^{i+1}\, b)\wedge \Eventually(\markTwo\wedge \Next^{i+1} b))\Bigr) \, \longrightarrow \,  \displaystyle{\bigvee_{b\in \{0,1\}}}((\Next\, b)\wedge \Eventually(\markTwo\wedge \Next b))
\]
\vspace{-0.2cm}
This concludes the proof of Lemma~\ref{lemma:constructionOfATLStarFOrmula}.
\end{proof}


%% file: FigureLowerBound.tex
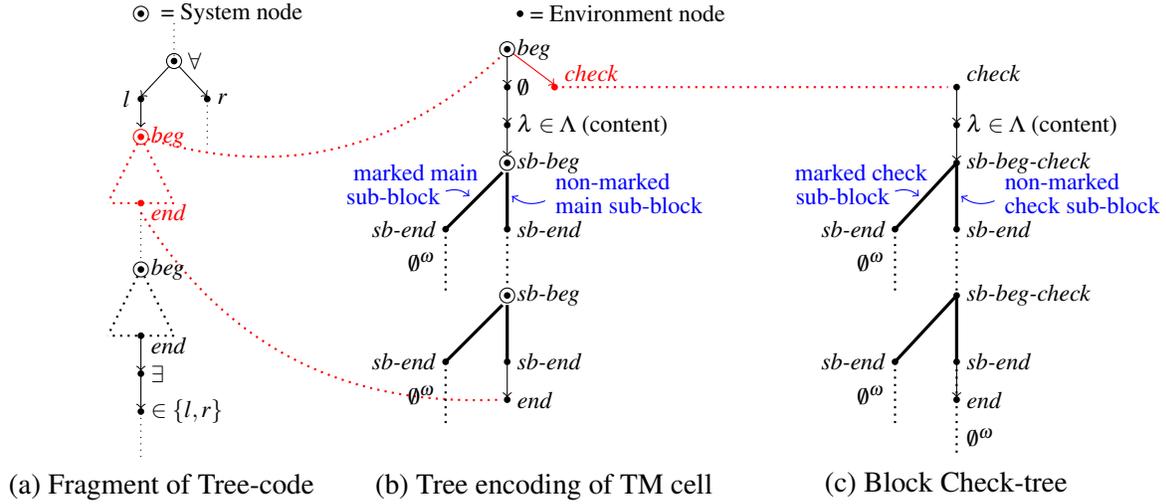
\begin{figure}[tb]
  \centering
  \caption{Encoding of computation trees of $\mathcal{M}$}\label{FigureLowerBound}
  \begin{minipage}{.3\textwidth}
\begin{tikzpicture}[scale=0.63]
\draw[draw=none,use as bounding box](-2.5,2) rectangle (3.3,-8.5);

\coordinate [label=right:{\footnotesize  \, =  System node}] (SystemNode) at (0,1.0);
\node[shape=circle,draw=black,inner sep=2pt,fill=white](SN) at (0,1.0) {};
\fill [black] (SystemNode) circle (2pt);

\coordinate [label=right:{\footnotesize  =   Environment node}] (EnvironmentNode) at (8,1.0);
\fill [black] (EnvironmentNode) circle (2pt);

\coordinate [label=right:{\footnotesize \,$\forall$}] (UnivNode) at (0.7,0);
\coordinate [label=right:{\footnotesize $r$}] (RightSucc) at (1.4,-0.8);
\path[->, thin,black] (UnivNode) edge  (0,-0.75);
\path[->, thin,black] (UnivNode) edge  (1.4,-0.75);

\node[shape=circle,draw=black,inner sep=2pt,fill=white](A) at (0.7,0) {};
\draw[thin, dotted, black] (RightSucc) edge (1.4,-1.8);
\draw[thin, dotted, black] (0.7,0.8) edge (UnivNode);

\coordinate [label=left:{\footnotesize $l$}] (LeftSucc) at (0,-0.8);
\fill [black] (LeftSucc) circle (2pt);
\fill [black] (RightSucc) circle (2pt);
\fill [black] (UnivNode) circle (2pt);
\path[->, thin,black] ( LeftSucc) edge  (0,-1.4);

\coordinate [label=right:\textcolor{red}{\footnotesize $\begOne$}] (Beg) at (0,-1.6);
\coordinate [label=right:\textcolor{red}{\footnotesize $\EndOne$}] (End) at (0,-3.2);

\node[shape=circle,draw=red,inner sep=2pt,fill=white](B) at (0,-1.6) {};

\draw(Beg.east) edge[thick,dotted, bend right,red] (7.6,0.1);
\draw(End.east) edge[thick,dotted, bend right,red] (7.6,-7.15);

\fill [red] (Beg) circle (2pt);
\fill [red] (0,-3.0) circle (2pt);
\draw[thick, dotted, red] (Beg) edge (0.7,-3.0);
\draw[thick, dotted, red] (Beg) edge (-0.7,-3.0);
\draw[thick, dotted, red] (0.7,-3.0) edge (-0.7,-3.0);

\coordinate [label=right:{\footnotesize $\begOne$}] (BegLast) at (0,-4.4);
\coordinate [label=right:{\footnotesize $\EndOne$}] (EndLast) at (0,-6.0);

\node[shape=circle,draw=black,inner sep=2pt,fill=white](B) at (0,-4.4) {};

\draw[thin, dotted, black] (End) edge (0,-4.2);

\fill [black] (BegLast) circle (2pt);
\fill [black] (0,-5.8) circle (2pt);
\draw[thick, dotted, black] (BegLast) edge (0.7,-5.8);
\draw[thick, dotted, black] (BegLast) edge (-0.7,-5.8);
\draw[thick, dotted, black] (0.7,-5.8) edge (-0.7,-5.8);

\coordinate [label=right:{\footnotesize $\exists$}] (ExistsNode) at (0,-6.6);
\coordinate [label=right:{\footnotesize $\in\{l,r\}$}] (TypeNode) at (0,-7.4);

\path[->, thin,black] (0,-5.8) edge  (0,-6.54);
\path[->, thin,black] (ExistsNode) edge  (0,-7.34);
\fill [black] (ExistsNode) circle (2pt);
\fill [black] (TypeNode) circle (2pt);
\draw[thin, dotted, black] (TypeNode) edge (0,-8.4);

\coordinate [label=right:{(a) Fragment of Tree-code}] (Caption) at (-3,-9);
\end{tikzpicture}

  \end{minipage}\quad
\begin{minipage}{.31\textwidth}
\begin{tikzpicture}[scale=0.63]
\draw[draw=none,use as bounding box](-2,1.5) rectangle (3.3,-8.5);

\coordinate [label=right:{\footnotesize $\begOne$}] (BegCheck) at (0,0);
\coordinate [label=right:{\footnotesize $\emptyset$}] (Main) at (0,-0.8);
\draw[thick,dotted, red] (1,-0.8) edge (9.3,-0.8);
\coordinate [label=right:\textcolor{red}{\footnotesize $\checkOne$}] (Check) at (1,-0.5);
\path[->, thin,black] (BegCheck) edge  (0,-0.75);
\path[->, thin,red] (BegCheck) edge  (1,-0.75);

\node[shape=circle,draw=black,inner sep=2pt,fill=white](A) at (0,0) {};
\fill [black] (BegCheck) circle (2pt);

\fill [black] (Main) circle (2pt);
\fill [black,red] (1,-0.8) circle (2pt);

\path[->, thin,black] (Main) edge  (0,-1.55);

\coordinate [label=right:{\footnotesize $\lambda\in\Lambda$ (content)}] (ContentCheck) at (0,-1.6);
\coordinate [label=right:{\footnotesize $\begTwo$}] (SBegCheck) at (0,-2.4);
\coordinate [label=right:{\footnotesize $\EndTwo$}] (SEndCheck) at (0,-3.8);
\coordinate [label=left:{\footnotesize $\EndTwo$}] (SEndCheckMark) at (-1.3,-3.8);

\path[->, thin,black] (ContentCheck) edge  (0,-2.25);
\fill [black] (ContentCheck) circle (2pt);
\node[shape=circle,draw=black,inner sep=2pt,fill=white](B) at (0,-2.4) {};
\fill [black] (SBegCheck) circle (2pt);

\fill [black] (SEndCheck) circle (2pt);
\fill [black] (SEndCheckMark) circle (2pt);
\draw[very thick, black] (0,-2.6) edge (SEndCheck);
\draw[very thick, black] (-0.1,-2.6) edge (SEndCheckMark);

\coordinate  (CommRInit) at (0.1,-3.1);
\coordinate  (CommR) at (0.8,-3.2);
\coordinate [label=right:\textcolor{blue}{\footnotesize non-marked}] (LineAR) at (0.8,-2.8);
\coordinate [label=right:\textcolor{blue}{\footnotesize main sub-block}] (LineBR) at (0.8,-3.3);
\draw(CommR.west) edge[->, bend left,blue] (CommRInit);

\coordinate  (CommLInit) at (-0.8,-3.1);
\coordinate  (CommL) at (-1.3,-3.0);
\coordinate [label=left:\textcolor{blue}{\footnotesize marked main}] (LineAL) at (-0.4,-2.6);
\coordinate [label=left:\textcolor{blue}{\footnotesize sub-block}] (LineBL) at (-1.2,-3.1);
\draw(CommL.east) edge[->, bend left,blue] (CommLInit);

\coordinate [label=right:{\footnotesize $\begTwo$}] (SBegCheckLast) at (0,-5.2);
\coordinate [label=right:{\footnotesize $\EndTwo$}] (SEndCheckLast) at (0,-6.6);
\coordinate [label=left:{\footnotesize $\EndTwo$}] (SEndCheckMarkLast) at (-1.3,-6.6);

\coordinate [label=right:{\footnotesize $\EndOne$}] (EndCheck) at (0,-7.4);

\draw[thick,dotted, black] (SEndCheck) edge (SBegCheckLast);
\path[->, thin,black] (SEndCheckLast) edge  (0,-7.35);
\node[shape=circle,draw=black,inner sep=2pt,fill=white](B) at (0,-5.2) {};
\fill [black] (SBegCheckLast) circle (2pt);

\fill [black] (SEndCheckLast) circle (2pt);
\fill [black] (SEndCheckMarkLast) circle (2pt);
\fill [black] (EndCheck) circle (2pt);
\draw[very thick, black] (0,-5.4) edge (SEndCheckLast);
\draw[very thick, black] (-0.1,-5.4) edge (SEndCheckMarkLast);

\draw[thick,dotted, black] (SEndCheckMark) edge (-1.3,-5.2);
\draw[thick,dotted, black] (SEndCheckMarkLast) edge (-1.3,-8.0);

\coordinate [label=left:{\footnotesize $\emptyset^{\omega}$}] (SMarkCheckEmpty) at (-1.3,-4.5);
\coordinate [label=left:{\footnotesize $\emptyset^{\omega}$}] (SMarkCheckEmptyLast) at (-1.3,-7.3);

\coordinate [label=right:{(b) Tree encoding of TM cell}] (Caption) at (-3,-9.24);

\end{tikzpicture}
 \end{minipage}\quad
  \begin{minipage}{.31\textwidth}
\begin{tikzpicture}[scale=0.63]
\draw[draw=none,use as bounding box](-3,1.5) rectangle (3.3,-8.5);

\coordinate [label=right:{\footnotesize $\checkOne$}] (Check) at (0,-0.5);

\fill [black] (0,-0.8) circle (2pt);

\path[->, thin,black] (0,-0.8) edge  (0,-1.55);

\coordinate [label=right:{\footnotesize $\lambda\in\Lambda$ (content)}] (ContentCheck) at (0,-1.6);
\coordinate [label=right:{\footnotesize $\checkTwo$}] (SBegCheck) at (0,-2.4);
\coordinate [label=right:{\footnotesize $\EndTwo$}] (SEndCheck) at (0,-3.8);
\coordinate [label=left:{\footnotesize $\EndTwo$}] (SEndCheckMark) at (-1.3,-3.8);

\path[->, thin,black] (ContentCheck) edge  (0,-2.35);
\fill [black] (ContentCheck) circle (2pt);
\fill [black] (SBegCheck) circle (2pt);

\fill [black] (SEndCheck) circle (2pt);
\fill [black] (SEndCheckMark) circle (2pt);
\draw[very thick, black] (SBegCheck) edge (SEndCheck);
\draw[very thick, black] (SBegCheck) edge (SEndCheckMark);

\coordinate  (CommRInit) at (0.1,-3.1);
\coordinate  (CommR) at (0.8,-3.2);
\coordinate [label=right:\textcolor{blue}{\footnotesize non-marked}] (LineAR) at (0.8,-2.8);
\coordinate [label=right:\textcolor{blue}{\footnotesize check sub-block}] (LineBR) at (0.8,-3.3);
\draw(CommR.west) edge[->, bend left,blue] (CommRInit);

\coordinate  (CommLInit) at (-0.8,-3.1);
\coordinate  (CommL) at (-1.3,-3.0);
\coordinate [label=left:\textcolor{blue}{\footnotesize marked check}] (LineAL) at (-0.4,-2.6);
\coordinate [label=left:\textcolor{blue}{\footnotesize sub-block}] (LineBL) at (-1.2,-3.1);
\draw(CommL.east) edge[->, bend left,blue] (CommLInit);

\coordinate [label=right:{\footnotesize $\checkTwo$}] (SBegCheckLast) at (0,-5.2);
\coordinate [label=right:{\footnotesize $\EndTwo$}] (SEndCheckLast) at (0,-6.6);
\coordinate [label=left:{\footnotesize $\EndTwo$}] (SEndCheckMarkLast) at (-1.3,-6.6);

\coordinate [label=right:{\footnotesize $\EndOne$}] (EndCheck) at (0,-7.4);

\draw[thick,dotted, black] (SEndCheck) edge (SBegCheckLast);
\path[->, thin,black] (SEndCheckLast) edge  (0,-7.35);
\fill [black] (SBegCheckLast) circle (2pt);
\fill [black] (SEndCheckLast) circle (2pt);
\fill [black] (SEndCheckMarkLast) circle (2pt);
\fill [black] (EndCheck) circle (2pt);
\draw[very thick, black] (SBegCheckLast) edge (SEndCheckLast);
\draw[very thick, black] (SBegCheckLast) edge (SEndCheckMarkLast);

\draw[thick,dotted, black] (SEndCheckLast) edge (0,-8.6);
\draw[thick,dotted, black] (SEndCheckMark) edge (-1.3,-5.2);
\draw[thick,dotted, black] (SEndCheckMarkLast) edge (-1.3,-8.0);

\coordinate [label=right:{\footnotesize $\emptyset^{\omega}$}] (CheckTreeEmpty) at (0,-8.2);
\coordinate [label=left:{\footnotesize $\emptyset^{\omega}$}] (SMarkCheckEmpty) at (-1.3,-4.5);
\coordinate [label=left:{\footnotesize $\emptyset^{\omega}$}] (SMarkCheckEmptyLast) at (-1.3,-7.3);
\coordinate [label=right:{(c) Block Check-tree}] (Caption) at (-3,-9.19);

\end{tikzpicture}
  \end{minipage}
\end{figure}